%% file: BCP.tex
\documentclass[a4paper,11pt]{article}

\usepackage{booktabs} 
\usepackage{colortbl} 
\usepackage[table]{xcolor} 
\usepackage{xfrac}
\usepackage[utf8]{inputenc}
\usepackage{amsmath,amsthm, amssymb, amsfonts, amsbsy}
\usepackage[mathscr]{euscript}
\usepackage{xspace}
\usepackage{graphicx,subfigure}
\usepackage{color}
\usepackage{enumerate}
\usepackage[english]{babel}
\usepackage{cite}
\usepackage{verbatim}
\usepackage[margin=1.2in]{geometry}
\usepackage{authblk}
\usepackage{hyperref}
\usepackage{float}
\usepackage{soul}
\usepackage{rotating}
\usepackage{caption}
\usepackage[draft]{fixme}
\usepackage{longtable}
\fxsetup{
  layout=marginnote,
  marginface=\normalfont\tiny,
  envface=,
  inlineface=,
  innerlayout=noinline,
}

\definecolor{lightgray}{gray}{0.9}
\restylefloat{table}

\usepackage{algorithm}
\usepackage[noend]{algpseudocode}

\newcommand{\R}{\mathbb{R}}

\newcommand{\B}{\mathbb{B}}
\newcommand{\Q}{\mathbb{Q}}

\newcommand{\NP}{\mathscr{NP}}

\newcommand{\bigO}{\mathcal{O}}
\newcommand{\polytope}{\mathcal{P}}

 \let\mathscr\relax
\usepackage[scr]{rsfso}

\newcommand{\BCPk}{\textsc{BCP}_k}
\newcommand{\BCP}{\textsc{BCP}}
\newcommand{\BCPtwo}{\textsc{BCP}_2}
\newcommand{\BCPkUnw}{\textsc{1-BCP}_k}
\newcommand{\BCPtwoUnw}{\textsc{1-BCP}_2}
\newcommand{\CutForm}{\textsc{cut-alg}}   
\newcommand{\FlowForm}{\textsc{flow-alg}}  
\newcommand{\FlowTwoForm}{\textsc{flow2-alg}} 

\providecommand{\keywords}[1]{\textit{Keywords:} #1}

\DeclareMathOperator{\opt}{\rm{OPT}}
\DeclareMathOperator{\convexhull}{\rm{convex-hull}}

\theoremstyle{plain}
\newtheorem{theorem}{Theorem}
\newtheorem{proposition}[theorem]{Proposition}

\theoremstyle{definition}
\newtheorem{problem}{Problem}

\usepackage{silence}
\begin{document}

\title{Integer Programming Approaches to \\ Balanced Connected $k$-Partition\protect\footnote{Research partially supported by grant \#2015/11937-9, S\~ao Paulo Research Foundation (FAPESP).}}

\author[1]{Flávio K. Miyazawa\thanks{(fkm@ic.unicamp.br) supported by
    CNPq (Proc. 314366/2018-0 and 425340/2016-3)  and FAPESP (Proc. 2016/01860-1)}}
\author[2]{Phablo F. S. Moura\thanks{(phablo@ic.unicamp.br) supported by FAPESP  grants 2016/21250-3 and 2017/22611-2, and CAPES.}}
\author[3]{Matheus J. Ota\thanks{(matheus.ota@students.ic.unicamp.br)  supported by CNPq.}}
\affil[1]{Instituto de Computação, Universidade Estadual de Campinas, Brazil}
\author[4]{Yoshiko Wakabayashi\thanks{(yw@ime.usp.br) supported by CNPq (Proc. 306464/2016-0 and 423833/2018-9).}} 
\affil[2]{Instituto de Matemática e Estatística, Universidade de São Paulo, Brazil}

\maketitle

\begin{abstract}
We address the problem of partitioning a vertex-weighted connected graph into~$k$ connected
  subgraphs that have similar weights, for a fixed integer~$k\geq 2$. This problem, known as the
  \textit{balanced connected $k$-partition problem} ($\BCPk$), is defined as follows.
  Given a connected graph~$G$ with nonnegative weights on the vertices, find a
  partition~$\{V_i\}_{i=1}^k$ of~$V(G)$ such that each class $V_i$ induces a connected subgraph
  of~$G$, and the weight of a class with the minimum weight is as large as possible.  
  It is known  that~$\BCPk$ is $\NP$-hard even on bipartite graphs and on interval graphs.  
  It has been largely investigated under different approaches and perspectives.
  On the practical side, $\BCPk$ is used to model many applications
  arising in police patrolling, image processing, cluster analysis, operating systems and robotics.
  We propose three integer linear programming formulations for the balanced connected $k$-partition problem.
  The first one contains only binary variables and a potentially large number of constraints that are separable in polynomial time. 
  Some polyhedral results on this formulation, when all vertices have unit weight, are also presented.
  The other formulations are based on flows and have a polynomial number of constraints and variables.
  Preliminary computational experiments have shown that the proposed formulations
  outperform the other formulations presented in the literature.
\end{abstract}


\input{intro-new}
\input{cut-form-new}
\input{flow-form-new}
\input{experiments_tmp}

\input{conclusion}

\bibliographystyle{abbrv}
\bibliography{BCP-bibliography}

\end{document}

%% file: intro-new.tex
\section{Introduction}
 
The graphs considered here are simple, connected and undirected. 
If~$G$ is a graph, then $V(G)$ denotes its vertex set and $E(G)$ denotes its edge set.
Later, for some formulations, we shall refer to directed graphs (or simply, digraphs)~$D$, with vertex set $V(D)$ and arc set $A(D)$.  
For an integer~$k\geq 1$, as usual, the symbol~$[k]$ denotes the set $\{1,2,\dots,k\}$.
Throughout this text, $k$ denotes a positive integer number.

Let $G$ be a connected graph.  A \emph{$k$-partition} of~$G$ is a
collection~$\{V_i\}_{i \in [k]}$ of nonempty subsets of~$V(G)$ such
that~$\bigcup_{i=1}^k V_i = V(G)$, and \hbox{$V_i \cap V_j = \emptyset$} for all
$i,j \in [k]$, $i \neq j$.  We refer to each set $V_i$ as a \emph{class} of the partition.
In this context, we assume that $|V(G)|\geq k$, otherwise~$G$ does not admit a
$k$-partition.  We say that a $k$-partition $\{V_i\}_{i\in [k]}$ of $G$ is
\emph{connected} if~$G[V_i]$, the subgraph of~$G$ induced by~$V_i$, is connected for
each~$i \in [k]$.
Let~$w\colon V(G) \to \Q_>$ be a function that assigns weights to the vertices of~$G$.
For every subset~$V^\prime \subseteq V(G)$, we define~$w(V^\prime)= \sum_{v \in V^\prime} w(v)$.
To simplify notation, we write~$w(H)=w(V(H))$ when~$H$ is a graph.
In the \emph{balanced connected $k$-partition problem} ($\BCPk$), we are given a
vertex-weighted connected graph, and we seek a connected $k$-partition such that the
weight of a lightest class of this partition is maximized. A more formal definition of
this problem is given below, as well as an example of an instance for $\BCPtwo$ (see
Figure~\ref{fig:instance-sol}).
For some fixed positive integer~$k$, each input of~$\BCPk$ is given by a pair~$(G,w)$.  We denote by~$\opt_k(G,w)$ the weight of a lightest set
in an optimal connected $k$-partition of~$V(G)$; but write simply $\opt(G,w)$ when~$k$ is
clear from the context.  Furthermore, we simply denote by~$G$ the instance in which~$w$ is
an assignment of a constant value to all vertices (which we may assume to
be~$1$). 

\begin{problem}
 \textsc{Balanced Connected $k$-Partition} ($\BCPk$)\\
 \textsc{Instance:} a connected graph $G$ and a vertex-weight function $w \colon V(G) \to \Q_>$. \\
 \textsc{Find:} a connected $k$-partition $\{V_i\}_{i\in [k]}$ of $V(G)$.\\
 \textsc{Goal:} maximize $\min_{i \in [k]} (w(V_i))$.
\end{problem}

\begin{figure}[!hbt]
    \centering
    \subfigure[A vertex-weighted graph.\label{subfig:instance}]{\includegraphics[scale=1.35]{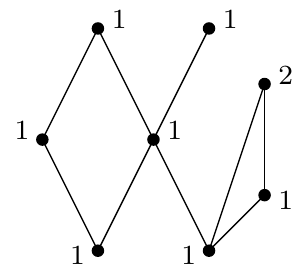}}
 \hspace{2cm}
    \subfigure[Optimal solution of cost 4. The dashed edge links vertices of different classes.]{\includegraphics[scale=1.35]{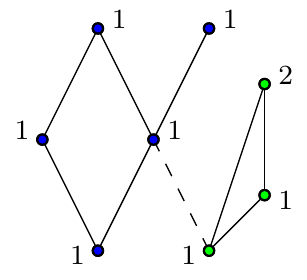}}
    \caption{An instance for~$\BCPtwo$ and its  optimal solution\label{fig:instance-sol}}
\end{figure}

There are several problems in police patrolling, image processing, data base, operating
systems, cluster analysis, education and robotics that can be modeled as a balanced
connected partition problem~\cite{BecPer83, LucPerSim89, LucPerSim93, MarSimNal97,
  MatBoz12, ZhoWanDinHuSha19}.  These different real-world applications indicate the
importance of designing algorithms for~$\BCPk$, and reporting on the computational
experiments with their implementations.  Not less important are the theoretical studies of
the rich and diverse mathematical formulations and the polyhedral investigations $\BCPk$
leads to.

Let us denote by~$\BCPkUnw$ the restricted version of~$\BCPk$ in which all vertices have
unit weight.  One may easily check that~$\BCPtwoUnw$ on 2-connected graphs can be solved
in polynomial time.  This problem also admits polynomial-time algorithms on graphs such
that each block has at most two articulation vertices~\cite{AliCal99, Chl96, GalMorMaf95,
  GalMorMaf97}.  Dyer and Frieze~\cite{DyeFri85} showed that, for every~$k \geq 2$,
~$\BCPkUnw$ is $\NP$-hard on bipartite graphs.  Wu~\cite{Wu12} proved that $\BCPk$ is
$\NP$-hard even on interval graphs, for all~$k \geq 2$.

For the~$\BCPtwo$ (weighted version), Chlebíková~\cite{Chl96} showed that this problem
is $\NP$-hard to approximate within an absolute error guarantee
of~$|V(G)|^{1- \varepsilon}$, for all~$\varepsilon > 0$.  In that same paper, Chlebíková
designed a $4/3$-approximation algorithm for that problem.  For~$\BCP_3$ and $\BCP_4$ on
$3$-connected and~$4$-connected graphs, respectively, there exist $2$-approximation
algorithms proposed by Chataigner~et~al.~\cite{ChaSalWak07}.

Wu~\cite{Wu12} showed the first pseudo-polynomial algorithm for~$\BCPtwo$ restricted to
interval graphs.  Based on this algorithm and using a scaling technique, Wu obtained a
$(1+\varepsilon)$-approximation with running time~$\bigO((1+1/\varepsilon) n^3)$,
where~$n$ is the number of vertices of the input graph.

Bornd\"orfer~et~al.~\cite{BorEliSch19} designed a $\Delta$-approximation for~$\BCPk$,
where~$\Delta$ is the maximum degree of an arbitrary spanning tree of the given graph.
This is the first known approximation algorithm on general graphs.

Mixed integer linear programming (MILP) formulations were proposed for~$\BCPtwo$ by
Matic~\cite{Mat14} and for~$\BCPk$ by Zhou~et~al.~\cite{ZhoWanDinHuSha19}.  Additionally,
Matic presented an heuristic algorithm based on a variable neighborhood search (VNS)
technique for $\BCPtwo$, and Zhou~et~al. devised a genetic algorithm for~$\BCPk$.  In both
works, the authors showed results of computational experiments to illustrate the quality
of the solutions constructed by the proposed heuristics and their running times compared
to the exact algorithms based on the MILP formulations.  
No polyhedral study of their formulations was reported.

 This paper is organized as follows.
In Section~\ref{section:cut}, we present a natural cut formulation for~$\BCPk$ and also two stronger valid inequalities for this formulation.
A further polyhedral study  of this formulation, when all vertices have unit weight, is presented in Section~\ref{section:cut-polytope}. 
Two of the inequalities in the formulation are shown to define facets, and one of them is characterized when it is facet-defining. 
In Section~\ref{section:flow}, we present a flow and a multicommodity flow based formulations for~$\BCPk$.
In Section~\ref{section:experiments}, we report on the computational
experiments with our formulations and also with those presented by
Matic~\cite{Mat14} and Zhou et al.~\cite{ZhoWanDinHuSha19}. 
We summarize our theoretical and practical contributions for~$\BCPk$
in Section~\ref{section:conclusion}.

%% file: cut-form-new.tex
\section{Cut formulation}
\label{section:cut}

In this section, the following concept will be useful. Let~$u$ and~$v$ be two non-adjacent vertices
in a graph~$G$.  We say that a set~$S \subseteq V(G)\setminus\{u,v\}$ is a $(u,v)$-\emph{separator}
if $u$ and $v$ belong to different components of~$G-S$.  We define~$\Gamma(u,v)$ as the collection
of all minimal $(u,v)$-separators in $G$.

Let $(G,w)$ be an input for $\BCPk$.  We propose the following natural integer linear programming
formulation~$\mathcal{C}_k (G,w)$ for $\BCPk$.  For that, for every vertex~$v \in V(G)$
and~$i \in [k]$, we define a binary variable~$x_{v,i}$ representing that~$v$ belongs to the $i$-th
class if and only if~$x_{v,i}=1$. In this formulation, to get hold of a class with the smallest
weight, we impose an ordering of the classes, according to their weights, so that the first class
becomes the one whose weight we want to minimize.

\begin{align} 
           \max \:&  \sum_{v \in V(G)} w(v)\: x_{v,1}\nonumber & \\
    \text{s.t.} \:& \sum_{v \in V(G)} w(v)\:  x_{v,i}  \leq \sum_{v \in V(G)} w(v)\:  x_{v,i+1} & \!\!\! \!\!\!  \forall v \in V(G), i \in [k-1],\label{ineq:asymetric} \\ 
                \:& \sum_{i \in [k]} x_{v,i}  \leq 1 &  \!\!\! \!\!\! \forall v \in V(G),\label{ineq:cover} \\ 
                \:& x_{u,i} +  x_{v,i} - \sum_{z \in S} x_{z,i}  \leq 1 &  \!\!\! \!\!\! \forall  uv \notin E(G), S \in \Gamma(u,v), i \in [k], \label{ineq:cut} \\ 
                \:& x_{v,i} \in \{0,1\} & \!\!\! \!\!\! \forall v \in V(G) \text{ and } i \in [k].\label{ineq:integer} 
\end{align}

Inequalities~\eqref{ineq:asymetric} imply a non-decreasing weight ordering of the classes. 
Inequalities~\eqref{ineq:cover} impose that every vertex is assigned to at most one class.
Inequalities~\eqref{ineq:cut} guarantee that every class induces a connected subgraph.
The objective function maximizes the weight of the first class. Thus, in an optimal solution no class
will be empty, and therefore it will always correspond to a connected $k$-partition of $G$.

We observe that the separation problem associated with inequalities~\eqref{ineq:cut} can be solved in
polynomial time by reducing it to the minimum cut problem.
Thus, the linear relaxation of~$\mathcal{C}_k$ can be solved in polynomial time because of the equivalence of separation and optimization (see~\cite{grotschel2012geometric}).

Since the feasible solutions of the formulation above may have empty classes, to refer to these solutions
 we introduce the following concept. A \emph{connected $k$-subpartition of~$G$} is
 a connected $k$-partition of a subgraph (not necessarily proper) of~$G$.
 Henceforth, we assume that if ~$\mathcal{V}=\{V_i\}_{i=1}^k$ is a
 connected $k$-subpartition of~$G$, then $w(V_i)\leq w(V_{i+1})$ for all~$i \in [k-1]$.
 For such a $k$-subpartition $\mathcal{V}$, we denote by~$\xi(\mathcal{V}) \in \B^{nk}$ the binary vector such that its non-null entries are
 precisely~$\xi(\mathcal{V})_{v,i} = 1$ for all~$i \in [k]$ and~$v \in V_i$ (that is, $\xi(\mathcal{V})$ denotes the incidence vector
 of $\mathcal{V}$).

 We next show that the previous formulation correctly models $\BCPk$.  For that, let~$\polytope_k(G,w)$
 be the polytope associated with  that formulation, that is,
$$\polytope_k(G,w) = \convexhull \{x \in \B^{nk} \colon x \text{ satisfies inequalities }\eqref{ineq:asymetric}-\eqref{ineq:cut} \text{ of }  \mathcal{C}_k (G,w)\}.$$
In the next proposition we show that $\polytope_k(G,w)$ is the convex hull of the incidence vectors of connected
$k$-subpartitions of $G$. 

\begin{proposition}
  Let $(G,w)$ be an input for $\BCPk$. Then, the following holds.
$$\polytope_k(G,w) = \convexhull \{\xi(\mathcal{V}) \in \B^{nk} \colon \mathcal{V} \text{ is a connected $k$-subpartition of }G \}.$$
\end{proposition}
\begin{proof}
  Consider first an extreme point~$x \in \polytope_k(G,w)$.  For each~$i \in [k]$, we define the set
  of vertices~$U_i = \{v \in V(G) \colon x_{v,i}=1\}$.  It follows from
  inequalities~\eqref{ineq:asymetric} and~\eqref{ineq:cover} that~$\mathcal{U}:= \{U_i\}_{i=1}^k$ is
  a $k$-subpartition of~$G$ such that~$w(U_i) \leq w(U_{i+1})$ for all~$i \in [k-1]$.
 
  To prove that $\mathcal{U}$ is a connected $k$-subpartition, we suppose to the contrary that there
  exists~$i \in [k]$ such that $G[U_i]$ is not connected.  Hence, there exist vertices~$u$ and~$v$
  belonging to two different components of~$G[U_i]$.  Moreover, there is a minimal set of
  vertices~$S$ that separates~$u$ and~$v$ and such that~$S\cap U_i = \emptyset$.  This implies that
  $x_{v,i} + x_{u,i} - \sum_{z \in S} x_{z,i} = x_{v,i} + x_{u,i} = 2$, a contradiction to the fact that ~$x$
  satisfies inequalities~\eqref{ineq:cut}.  Therefore, $\mathcal{U}$ is a connected $k$-subpartition
  of~$G$.
 
  To show the converse, consider now a connected $k$-subpartition~$\mathcal{V}=\{V_i\}_{i=1}^k$ of~$G$.
  By the definition of~$\xi(\mathcal{V})$, it is clear that this vector satisfies
  inequalities~\eqref{ineq:asymetric} and~\eqref{ineq:cover}.  Take a fixed~$i \in [k]$.  For
  every pair ~$u$,$v$ of non-adjacent vertices in~$V_i$, and every $(u,v)$-separator~$S$ in~$G$, it holds
  that~$S\cap V_i\neq \emptyset$, because ~$G[V_i]$ is connected.  Therefore, $\xi(\mathcal{V})$
  satisfies inequalities~\eqref{ineq:cut}.  Consequently, $\xi(\mathcal{V})$ belongs
  to~$\polytope_k(G,w)$.
\end{proof}

In the remainder of this section we present two further classes of
valid inequalities for $\polytope_k(G,w)$ that strenghten the
formulation~$\mathcal{C}_k(G,w)$.  We start showing a class that
dominates the class of inequalities~\eqref{ineq:cut}.

\begin{proposition}\label{prop:stronger-cut}
  Let~$(G,w)$ be an input for $\BCPk$.
  Let~$u$ and~$v$ be two non-adjacent vertices of~$G$, and let~$S$ be a minimal $(u,v)$-separator. 
  Let  $L= \left\{ z \in S \colon w(P_z)>\frac{w(G)}{k-i+1}\right\}$, where $P_z$ is a minimum-weight path between~$u$ and~$v$ in~$G$ that contains~$z$.
  For every~$i\in [k]$, the following inequality is valid for $\polytope_k(G,w)$:
\begin{equation}\label{ineq:asymetric:strong}
 x_{u,i} + x_{v,i} - \sum_{z \in S\setminus L} x_{z,i} \leq 1.
\end{equation}
\end{proposition}
\begin{proof}
Consider an extreme point~$x$ of $\polytope(G,w)$, and define~$V_i = \{v \in V(G) \colon x_{v,i}=1\}$ for each~$i \in [k]$.
Suppose to the contrary that there is~$ j \in [k]$ such that~$w(V_j)>\frac{w(G)}{k-j+1} $.
Since~$x$ satisfies inequalities~\eqref{ineq:asymetric}, it holds that $\sum_{i \in [k]\setminus [j-1]} w(V_i) > w(G)$, a contradiction.
Thus, if~$u$ and~$v$ belong to~$V_i$, then there exists a vertex~$z \in S\setminus L$ such that~$z$ also belongs to~$V_i$.
Therefore, $x$ satisfies inequality~\eqref{ineq:asymetric:strong}.
\end{proof}

Inspired by the inequalities devised by de~Aragão and~Uchoa~\cite{PogUch99} for a connected assignment problem, we propose the following class of inequalities for~$\polytope_k(G,w)$.

\begin{proposition}
  Let~$(G,w)$ be an input for $\BCPk$, and $q\geq 2$ be an integer.  Let $S$ be a subset of $V(G)$,
  $N(S)$ the set of neighbors of $S$, and ~$S'$ a subset of~$S$ containing~$q$ distinct pair of
  vertices $\{s_i,t_i\}$, $s_i\neq t_i$,  $i\in[q]$.  Moreover, let~$\sigma \colon [q] \to [k]$ be an injective
  function, and let~$I$ denote the image of~$\sigma$, that is,
  $I=\{\sigma(i) \in [k]\colon i \in [q]\}$.  If there is no collection of~$q$ vertex-disjoint
  $(s_i,t_i)$-paths in~$G[S]$, then the following inequality is valid
  for~$\polytope_k(G,w)$:
 \begin{equation}
   \sum_{i \in [q]} \left (x_{s_i, \sigma(i)} + x_{t_i, \sigma(i)} \right ) + \sum_{v \in N(S)} \sum_{i \in [k]\setminus I} x_{v,i} \leq 2q + |N(S)| -1 \enspace.\label{ineq:cross:general}
 \end{equation}
\end{proposition}
\begin{proof}
  Suppose, to the contrary, that there exists an extreme point~$x$ of~$\polytope_k(G,w)$ that
  violates inequality~\eqref{ineq:cross:general}.  Let
  $A = \sum_{i \in [q]} \left (x_{s_i, \sigma(i)} + x_{t_i, \sigma(i)} \right )$ and
  $B = \sum_{v \in N(S)} \sum_{i \in [k]\setminus I} x_{v,i}$.  From
  inequalities~\eqref{ineq:cover}, we have that~$A \leq 2q$. Since $x$ violates~\eqref{ineq:cross:general}, it follows
  that $B > |N(S)| - 1$.  Thus~$B = |N(S)|$ (because $x$ satisfies inequalities~\eqref{ineq:cover}).
  Hence, every vertex in $N(S)$ belongs to a class that is different from those indexed by~$I$.
  This implies that every class indexed by $I$ contains precisely one of the $q$ distinct pairs
  $\{s_i,t_i\}$. Therefore, there exists a collection of $q$~vertex-disjoint $(s_i,t_i)$-paths
  in~$G[S]$, a contradiction.
\end{proof}

Kawarabayashi et~al.~\cite{KawKobRee12} proved that, given an $n$-vertex graph~$G$ and a set of~$q$ pairs of terminals in~$G$, the problem of deciding whether~$G$ contains~$q$ vertex-disjoint paths linking the given pairs of terminals can be solved in time~$\bigO(n^2)$.
Hence, inequalities~\eqref{ineq:cross:general} can be separated in polynomial time when~$S=V(G)$.

\section{Polyhedral results for $\BCPkUnw$}
\label{section:cut-polytope}

In this section we focus on $\BCPkUnw$, the special case of $\BCPk$ in
which all vertices have unit weight. In this case, instead of
$\polytope_k(G,w)$, we simply write $\mathcal{P}_k(G)$, the polytope
defined as the convex hull of~$\{ x \in \B^{kn} : x \text{ satisfies } \eqref{ineq:asymetric}
\text{ - } \eqref{ineq:cut} \}$.

Note that, if the input graph $G$ has no matching of size~$k$, then $G$ has no feasible connected $k$-subpartition~$\{V_i\}_{i \in [k]}$ such that~$|V_i|\geq 2$ for all~$i \in [k]$, and thus~$\opt(G)=1$, and it is easy to find an optimal solution. Thus, we assume from now on that~$G$ has a matching of size~$k$ (a property that can be checked in polynomial time~\cite{edmonds1965paths}), and
that~$n\geq 2k$, where~$n:=|V(G)|$.

For each ~$v \in V(G)$ and~$i \in [k]$ we shall construct a binary vector~$\chi(v,i)$ that belongs to~$\mathcal{P}_k(G)$.
Let us denote by~$e(v,i) \in \B^{nk}$ the unit vector  such that its single non-null entry is indexed by~$v$ and~$i$.
Now consider any set $S \subseteq V(G)\setminus\{v\}$, $|S|=k-i$, and  a
bijective function~$\nu \colon S \to [k]\setminus [i]$. Since
$n\geq 2k$, such a set  $S$ exists.  Fix a pair $(S,\nu)$, where $S$ and $\nu$ are as
previously defined. 
Let~$\chi(v,i) \in \B^{nk}$ be the vector $e(v,i)+ \sum_{u \in S} e(u,\nu(u))$.  Note that
$\chi(v,i)$ belongs to~$\polytope_k(G)$,  it is the incidence vector of a $k$-subpartition, say
${\mathcal S}_i=\{S_i, \ldots, S_k\}$ of $G$, in which $v$ belongs to the class $S_i$, and each
vertex of $S\subseteq V(G)\setminus \{v\}$ belongs to one of the classes $S_{i+1},\ldots,S_k$, all
of which are singletons. 

  To be rigorous, we should write $\chi^{s,\nu}(v,i)$ as different choices of $S$ and $\nu$ give rise to different vectors, but we simply write $\chi(v,i)$ with the understanding that it refers to some $S$ and bijection $\nu$.

\begin{proposition}\label{prop:full-dimensional}
$\polytope_k(G)$ is full-dimensional, that is, $\dim(\polytope_k(G))=kn$.
\end{proposition}
\begin{proof}
  Let~$X$ be the set of~$kn$ vectors previously defined, that is,
  ~$X=\{\chi(v,i) \in \B^{nk} \colon v \in V(G) \text{ and } i \in [k]\}$.  Assume that
  $V(G)=\{v_1, \ldots, v_n\}$.  We suppose, with no loss of generality, that the indices of a vector
  $x$ in~$X$ are ordered as~$(x_{v_1,1}, \ldots, x_{v_n,1}, \ldots, x_{v_1,k}, \ldots, x_{v_n,k})$.

  Let $M$ be the matrix whose columns are precisely the $nk$ vectors in $X$ (in the following
  order): $\chi(v_1,1), \ldots, \chi(v_n,1), \ldots, \chi(v_1,k), \ldots, \chi(v_n,k)$. One may
  easily check that~$M$ is a lower triangular square matrix of dimension~$kn$.  Note that the
  columns of~$M$ are precisely the vectors in~$X$.  Hence the vectors in~$X$ are linearly
  independent. 
Since $X \subseteq \polytope_k(G)$, we conclude that~$\dim(\polytope_k(G))=nk$, that is, $\polytope_k(G)$ is full-dimensional.
\end{proof}


In the forthcoming proofs, we have to refer to some connected
$k$-subpartitions of $G$, defined (not uniquely) in terms of distinct
vertices $u$, $v$ of $G$, and specific classes $i$, $j$, where
$i<j$. For that, we define a short notation to represent the incidence
vectors of these connected $k$-subpartitions. Given such~$u$, $v$,
and~$i$, $j$, choose two set of vertices~$S$ and~$\hat S$ in~$G$, both
of cardinality $k-i+1$, and
bijections~$\pi \colon S \to \{i, \ldots, k\}$
and~$ \hat \pi \colon \hat S \to \{i, \ldots, k\}$ such that

\begin{enumerate}[(i)]
\item $u \in S\cap \hat S$  and  $v \in S \setminus \hat S$;
 \item $\pi(u)= \hat \pi(u)=i$  and  $\pi(v)= j$.
\end{enumerate}

Let~$\phi(u,i,v,j)$ and $\psi(u,i,v)$ be vectors in $\{0,1\}^{nk}$ such
that their non-null entries are precisely:
$\phi(u,i,v,j)_{z, \pi(z)} = 1$ for every $z \in S$, and
$\psi(u,i,v)_{z, \hat \pi(z)} = 1$ for every \hbox{$z \in \hat S$.}  The
vectors~$\phi(u,i,v,j)$ and~$\psi(u,i,v)$ clearly belong
to~$\polytope_k(G)$. Moreover, note 
that~$\phi(u,i,v,j)_{u,i}=\phi(u,i,v,j)_{v,j}=\psi(u,i,v)_{u,i}=1$
and~$\psi(u,i,v)_{v,\ell}=0$ for all~$\ell \in [k]$.

\begin{proposition}\label{prop:facet:non-negative}
For every~$v \in V(G)$ and~$i \in [k]$, the inequality~$x_{v,i}\geq 0$ induces a facet of~$\polytope_k(G)$.
\end{proposition}
\begin{proof}
  Similarly to the proof of Proposition~\ref{prop:full-dimensional}, let~$X_1=\{\chi(v,j) \in \B^{nk} \colon j \in [k]\setminus \{i\}\}$. 
  Additionally, we define~$X_2 = \{\psi(u,j,v) \in \B^{nk} \colon u \in V(G)\setminus \{v\} \text{ and } j \in [k]\}$.
  Note that $|X_1 \cup X_2|= nk-1$.
  Since the null vector and all vectors in~$X_1 \cup X_2$ are all affinely independent, and they all belong to the face~$\{ x \in \B^{nk} \colon x_{v,i}=0 \}$, we conclude that the inequality~$x_{v,i}\geq 0$ induces a facet of~$\polytope_k(G)$.
 \end{proof}


In what follows, considering that the polytope $\polytope_k(G)$ is
full-dimensional, to prove that a face
$\hat F = \{ x \in \polytope_k(G) \colon \hat\lambda x = \hat
\lambda_0\}$ is a facet of~$\polytope_k(G)$, we show that if a
nontrivial face~$F=\{ x \in \polytope_k(G) \colon \lambda x = \lambda_0\}$
of~$\polytope_k(G)$ contains~$\hat F$, then there exists a real
positive constant~$c$ such that~$\lambda = c \hat \lambda$
and~$\lambda_0 = c \hat \lambda_0$.
 
\begin{proposition}
 For every~$v \in V(G)$, the inequality  $\sum_{i \in [k]} x_{v,i} \leq 1$ induces a facet of~$\polytope_k(G)$.
\end{proposition}
\begin{proof}
  Fix a vertex~$v\in V(G)$. Let $\hat F_v = \{ x \in \polytope_k(G) \colon \hat\lambda x = \hat \lambda_0\}$, where~$\hat \lambda x \leq \hat \lambda_0$ corresponds to~$\sum_{i \in [k]} x_{v,i} \leq 1$. Let~$F=\{ x \in \polytope_k(G) \colon \lambda x = \lambda_0\}$ be a nontrivial face of~$\polytope_k(G)$ such that $\hat F_v \subseteq F$.  We shall prove that~$\lambda_{v,i}= \lambda_0$ and~$\lambda_{u, i} = 0$ for every~$u \in V(G)\setminus\{v\}$ and~$i \in [k]$.

 Since $G$ is nontrivial and connected, it is easy to see that $G$ has a set of $n$ nested connected subgraphs $G_1, G_2, \ldots, G_n$  such that $G_1$ consists solely of the vertex $v$, each~$G_j \subset G_{j+1}$ for~$j=2,\ldots n-1$, and $G_n=G$. (It suffices to consider a spanning tree in $G$, and starting from $v$,  define the subsequent  subgraphs by adding  at each step an appropriate edge and vertex from this spanning tree.)

Consider the set of vectors~$A=\{e(G_j, k)\}_{j\in [n]}$, where~$e(G_j,k) = \sum_{u \in V(G_j)} e(u,k)$ for every~$j \in [n]$.
Since~$v \in V(G_j)$ for all~$j \in [n]$, it follows that~$A \subseteq \hat F_v$.
As a consequence,~$\lambda_{u,k}=0$ for all~$u \in V(G)\setminus\{v\}$.
Additionally,~$\lambda_{v,k}=\lambda_0$ since~$e(v,k) \in \hat F_v$.

Let~$\ell \in [k]\setminus\{1\}$.  Suppose
that~$\lambda_{v,i}=\lambda_0$ and~$\lambda_{u,i}=0$ for
every~$u \in V(G)\setminus\{v\}$ and~$i \in \{\ell, \ldots, k\}$. Now define the set of
vectors~$B=\{\phi(u,\ell-1,v,k)\colon u \in V(G)\setminus\{v\}\}$.
Note that~$B\subseteq \hat F_v$,  since~$v$ belongs to exactly one
class of the partition corresponding to each vector in~$B$.
By the induction hypothesis, we obtain~$\lambda_{u,\ell-1}=0$ for each~$u \in V(G)\setminus\{v\}$.
Moreover, observe that $\chi(v,\ell-1)$ belongs to~$\hat F_v$.
It follows from the induction hypothesis that~$\lambda (\chi(v,\ell-1))= \lambda_{v,\ell-1} =\lambda_0$.

Therefore, we conclude that~$\lambda = \lambda \left(\sum_{i \in [k]} e(v,i) \right) = \lambda_0 \hat \lambda$.
Since $\lambda_0 \neq 0$ (otherwise $F$ would be a trivial face), it follows that $\lambda x \leq \lambda_0$ is a multiple scalar of
$\hat \lambda x \leq \hat \lambda_0$, and therefore~$\hat F_v$ is a facet of~$\polytope_k(G)$.
\end{proof}

Let~$u$ and~$v$ be two non-adjacent vertices of~$G$ and let~$S$ be a minimal
$(u,v)$-separator in~$G$.  We denote by~$H_u$ and~$H_v$ the components of~$G-S$ which
contain~$u$ and~$v$, respectively.  Since $S$ is minimal, it follows that every vertex
in~$S$ has at least one neighbor in~$H_u$ and one in~$H_v$.

In this context of minimal $(u,v)$-separator~$S$, for
every~$z \in V(G)$, the following two concepts (and notation) will be
important for the next result. 

We denote by~$G_z$ a minimum size connected subgraph of~$G$
containing~$z$, with the following properties: If~$z \in V(H_v)$, then
$G_z$ is contained in $H_v$; if~$z \in V(H_u)$, then $G_z$ is
contained in $H_u$.  Otherwise, $G_z$ contains~$u$,~$v$ and exactly
one vertex of~$S$, that is,~$|V(G_z) \cap S|=1$.  Clearly, such a
subgraph always exists. Moreover, if $z\notin S\cup\{u,v\}$, then the
subgraph $G_z- z$ is connected (that is, $z$ is not a cutvertex of
$G_z$).

For any  integer~$ i \in [k]$, we say that~$G$ admits a \emph{$(u,v,S,i)$-robust
   subpartition} if, for each~$z \in V(G)\setminus \{u,v\}$, there is a connected
 $(k-i)$-subpartition~$\{V_j\}_{j \in [k]\setminus [i]}$ of~$G-G_z$ such that
 $|V(G_z)|\leq |V_j|$ for all~$j \in [k] \setminus [i]$.

\begin{theorem}
Let~$u$ and~$v$ be non-adjacent vertices in~$G$, let~$S$ be a minimal $(u,v)$-separator, and let~$i \in [k]$.
Then $G$ admits a $(u,v,S,i)$-robust subpartition  if and only if the following inequality induces a facet of~$\polytope_k(G)$:
\[ x_{u,i} + x_{v,i} - \sum_{s \in S} x_{s,i}  \leq 1 \enspace.\]
\end{theorem}


%% file: flow-form-new.tex
\section{Flow formulations}
\label{section:flow}

We present in this section a mixed integer linear programming
formulation for~$\BCPk$ based on flows in a digraph. For that, given
an input $(G,w)$ for $\BCPk$, we construct a digraph~$D$ as follows: 
First, we add to~$G$ a set~$S=\{s_1, \ldots, s_k\}$ of $k$ new
vertices.  Then, we replace every edge of~$G$ with two arcs with the
same endpoints and opposite directions.  Finally, we add an arc from
each vertex in~$S$ to each vertex of~$G$ (see Figure \ref{fig:instance-trans}). More
formally, the vertex set of~$D$ is $V(D) = V(G) \cup S$ and its arc
set is
\[A(D) = \{(u,v), (v,u) : \{u,v\} \in E(G)\} \cup \{(s_i,v) : i \in [k], v \in V(G)\}.\]

Now, the idea behind the formulation is the following: find in~$D$ a
maximum flow from the~$k$ sources in~$S$ such that every vertex~$v$
in~$V(D)\setminus S$ receives flow only from a single vertex of~$D$
and consumes $w(v)$ of the received flow.  As we shall see, for
every~$i \in [k]$, the flow sent from source~$s_i$ corresponds to the
total weight of the vertices in the $i$-th class of the desired
partition.

To model this concept, with each arc~$a \in A(D)$, we associate a
non-negative real variable~$f_a$ that represents the amount of flow
passing through~$a$, and a binary variable~$y_{a}$ that equals one if
and only if arc~$a$ is used to transport a positive flow. The corresponding
formulation, shown below, is denoted $\mathcal{F}_k (G,w)$.

\begin{align} 
           \max \:&  \sum_{a \in \delta^+(s_1)} f_{a}\nonumber & \\
    \text{s.t.} \:& \sum_{a \in \delta^+(s_i)} f_{a}  \leq \sum_{a \in \delta^+(s_{i+1})} f_{a} & \!\!\!\!\!\!\!\!\!\!\!\!\!\!\!\!\!\!\! \forall i \in [k-1],\label{flow2:ineq:asym} \\  
                \:& \sum_{a \in \delta^{-}(v)} f_{a} - \sum_{a \in \delta^+(v)} f_{a} = w(v) &  \!\!\! \!\!\! \forall  v \in V(D) \setminus S, \label{flow2:ineq:conservation} \\ 
                \:& f_{a} \leq w(G) y_{a}  &  \!\!\!\!\!\! \forall a \in A(D), \label{flow2:ineq:capacity}\\ 
                \:& \sum_{a \in \delta^+(s_i)} y_{a} \leq 1 &  \!\!\!\!\!\! \forall i \in [k],\label{flow2:ineq:outdeg} \\ 
                \:& \sum_{a \in \delta^-(v)} y_{a} \leq 1 &  \!\!\!\!\!\! \forall v \in V(D)\setminus S,\label{flow2:ineq:indeg}  \\ 
                \:& y_{a} \in \{0,1\} & \!\!\! \!\!\! \forall a \in A(D) \label{flow2:ineq:int} \\ 
                \:& f_{a} \in \R_\geq & \!\!\! \!\!\! \forall a \in A(D) \label{flow2:ineq:real}
\end{align}

Inequalities~\eqref{flow2:ineq:outdeg} impose that from every
source~$s_i$ at most one arc leaving it transports a positive flow
to a single vertex in~$V(D)\setminus S$.
Inequalities~\eqref{flow2:ineq:indeg} require that every non-source
vertex receives a positive flow from at most one vertex of~$D$.
By inequalities~\eqref{flow2:ineq:capacity}, a positive flow can
only pass through arcs that are chosen (arcs $a$ for which $y(a)=1$).
Inequalities~\eqref{flow2:ineq:conservation} guarantee that each
vertex~$v \in V(D)\setminus S$ consumes~$w(v)$ of the flow that it
receives.  Finally, inequalities~\eqref{flow2:ineq:asym} impose that
the amount of flow sent by the sources are in a non-decreasing order.
This explains the objective function.

\begin{figure}[!hbt]
  \centering \subfigure[The digraph $D$ obtained from the graph $G$ shown in
  Figure~\ref{subfig:instance}. Vertices~$s_1$ and $s_2$ dominates all
  vertices in the dashed circle. The numbers on the vertices are the
  weights.]{\includegraphics[scale=1.5]{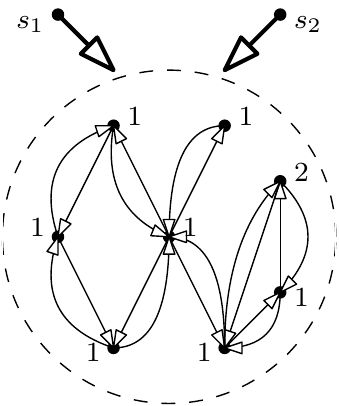}}
  \qquad \quad \subfigure[A feasible solution: the arcs represent
  non-zero $y$-variables, and the flow in  each of them is indicated on their side. Compare with the solution shown in Figure~\ref{fig:instance-sol}.]{\includegraphics[scale=1.5]{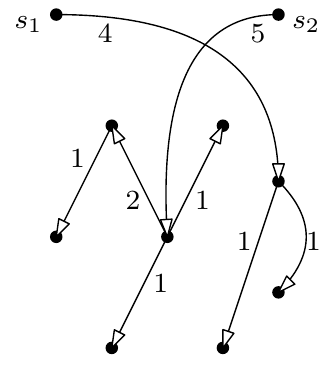}}
    \caption{Digraph $D$ and a feasible solution for  formulation~$\mathcal{F}_2$.\label{fig:instance-trans}}
\end{figure}

Since each non-source vertex receives flow from at most one vertex,
the flows sent by any two distinct sources do not pass through a same
vertex. That is, if a source $s_i$ sends an amount of flow, say $w_i$,
this amount $w_i$ is distributed to a subset of vertices, say $V_i$
(with total weight $w_i$); and all subsets $V_i$ are mutually
disjoint.
Moreover, $w_i$ is exactly the sum of the weights of the
vertices that receive flow from~$s_i$, and $G[V_i]$ is a connected
subgraph of $G$. (See Figure~\ref{fig:instance-trans}.)
It follows from these remarks that formulation~$\mathcal{F}_k$
correctly models~$\BCPk$.

The proposed formulation~$\mathcal{F}_k(G,w)$ has a total of $2nk+4m$
variables (half of them binary), and only~$\bigO(n+m+k)$ constraints,
where $n := |V(G)|$ and~$m := |E(G)|$.  The possible drawbacks of this
formulation are the large amount of symmetric solutions and the
dependency of inequalities~\eqref{flow2:ineq:capacity} on the weights
assigned to the vertices.  To overcome such disadvantages, we propose
in the next section another  model based on flows that considers a
total order of the vertices to avoid symmetries and uncouple the
weights assigned to the vertices from the flow circulating in the
digraph.

\subsection{A second flow formulation}

Our second formulation for~$\BCPk$, denoted by~$\mathcal{F}^\prime_k(G,w)$, is also based on a digraph~$D$ that is
constructed from~$G$ as follows.
It has vertex set $V(D) = V(G)\cup \{s\}$ and
arc set \(A(D) = \{(u,v), (v,u) \colon \{u,v\} \in E(G)\} \cup \{(s,v) \colon
v \in V(D)\}.\)

Moreover, it assumes that there is a total ordering~$\succ$ defined on
the vertices of~$G$.  
For simplicity, for a
vertex~$v \in V(D)\setminus\{s\}$ and integer $i \in [k]$, we use the
short notation $y(\delta^-(v),i)$ instead of $\sum_{a \in \delta^-(v)} y_{a,i}$.

\begin{align} 
           \max \:&  \sum_{v \in V(D)}  w(v)\: y(\delta^-(v),1) \nonumber & \\
    \text{s.t.} \:& \sum_{v \in V(D)\setminus\{s\}}  w(v)\: y(\delta^-(v),i) \leq \sum_{v \in V(D)\setminus\{s\}} w(v) \:y(\delta^-(v),i+1) \!\!\!\!\!\!\!\!\!\!\!\!\!\!\!\!\!\!\!\!\!\!\!\!\!\!\!\!\!\!\!\!\!\!\!\!\!\!\!\!\!\!\!\! & \forall i \in [k-1],\label{flow:ineq:asym}\\ 
                \:& \sum_{a \in \delta^+(s)} y_{a,i} \leq 1 &  \forall i \in [k], \label{flow:ineq:source}\\
                \:& \sum_{i \in [k]} y(\delta^-(v), i) \leq 1 &  \forall v \in V(D)\setminus\{s\}, \label{flow:ineq:indeg}\\ 
                \:&  y_{sv,i} + y(\delta^-(u),i) \leq 1 &  \forall u,v \in V(D)\setminus\{s\}, v \succ u, i \in [k], \label{flow:ineq:root}\\ 
                \:& f_{a,i} \leq n \: y_{a,i}&    \forall a \in A(D), i \in [k], \label{flow:ineq:capacity}\\ 
                \:& \sum_{a \in \delta^+(v)} f_{a,i} \leq \sum_{a \in \delta^-(v)} f_{a,i} &  \forall v \in V(D)\setminus\{s\}, i \in [k],\label{flow:ineq:conservation2} \\ 
                \:& \sum_{i \in [k]} \sum_{a \in \delta^-(v)} f_{a,i} - \sum_{i \in [k]} \sum_{a \in \delta^+(v)} f_{a,i}=1&  \forall v \in V(D)\setminus\{s\}, \label{flow:ineq:conservation}\\ 
                \:& y_{a,i} \in \B & \forall a \in A(D), i \in[k], \label{flow:ineq:int}\\
                \:& f_{a,i} \in \R_\geq & \forall a \in A(D), i \in[k]. \label{flow:ineq:real}
\end{align}

\medskip

To show that the above formulation indeed models~$\BCPk$, let us consider the following polytope:
\[\mathcal{Q}_k(G,w) = \convexhull \{(y,f) \in \B^{(n+2m)k} \times \R^{(n+2m)k} \colon (y,f) \text{ satisfies  ineq.  } \eqref{flow:ineq:asym}-\eqref{flow:ineq:real}\}.\]

Let~$\mathcal{V}$ be a connected $k$-subpartition of~$G$ such
that~$w(V_i) \leq w(V_{i+1})$ for all~$i \in [k-1]$.  
Then, for  each integer~$i \in [k]$, there exists in~$D$ an
out-arborescence~$\vec T_i$ rooted at~$r_i$ such
that~$V(\vec T_i)= V_i$ and~$v\succ r_i $ for
all~$v \in V_i \setminus \{r_i\}$.
Now, let~$g_i$ be the function 
$g_i \colon A(\vec T_i)\cup\{(s,r_i)\} \to \R_\geq$ defined as
follows: $g_i((u,v)) = 1$ if~$v$ is a leaf of~$\vec T_i$, and
$g_i((u,v))=1 + \sum_{(v,z) \in A(\vec T_i)} g_i((v,z))$,  otherwise.
It follows from this definition that~$g_i((s,r_i)) = |V_i|$.

We now define vectors~$\rho(\mathcal{V}) \in \B^{(n+2m)k}$ and $\tau (\mathcal{V}) \in  \R^{(n+2m)k}$ such that, for every arc $a \in A(D)$ and~$i \in [k]$, we have

\begin{minipage}{0.5\linewidth}
\noindent
\begin{displaymath}
    \rho(\mathcal{V})_{a,i} = \left\{
        \begin{array}{lr}
        1, & \text{ if } a \in A(\vec T_i)\cup \{(s,r_i)\}\\
        0, & \text{ otherwise,}
        \end{array}
    \right.
\end{displaymath}
\end{minipage}
\begin{minipage}{0.5\linewidth}
\noindent
\begin{displaymath}
    \tau(\mathcal{V})_{a,i} = \left\{
        \begin{array}{lr}
        g_i(a), & \text{ if } a \in A(\vec T_i)\cup \{(s,r_i)\}\\
        0, & \text{ otherwise.}
        \end{array}
    \right.
\end{displaymath}
\end{minipage}

\medskip
We are now ready to  prove the claimed statement on $\mathcal{Q}_k(G,w)$,

\begin{proposition}
  The polytope $\mathcal{Q}_k(G,w)$ is precisely the polytope
  \[\convexhull\{ (\rho(\mathcal{V}), \tau(\mathcal{V})) \in \B^{(n+2m)k} \times \R^{(n+2m)k} \colon \mathcal{V} \text{ is a connected $k$-partition of }G \}.\]
\end{proposition}
\begin{proof}
  Let $(y,f)$ be an extreme point~of $\mathcal{Q}_k(G,w)$; and for every~$i \in [k]$, let
  $U_i = \{v \in V(G) \colon y(\delta^-(v),i) = 1\}$. 
  It follows from inequalities~\eqref{flow:ineq:indeg} that, for every
  vertex $v \in V(D)\setminus\{s\}$, at most one of the arcs entering
  it is chosen.  Observe that inequalities~\eqref{flow:ineq:capacity}
  force that a flow of type~$i$ can only pass through an arc of
  type~$i$ if this arc is chosen.  Hence, every vertex receives at
  most one type of flow from its in-neighbors.  Furthermore,
  inequalities~\eqref{flow:ineq:conservation2}
  and~\eqref{flow:ineq:conservation} guarantee that the flow that
  enters a vertex and leaves it are of the same type, and that each
  vertex consumes exactly one unit of such flow.

  Inequalities~\eqref{flow:ineq:source} imply that all flow of a given
  type passes through at most one arc that has tail at the source~$s$.
  Therefore, we have that~$\{U_i\}_{i \in [k]}$ is a connected
  $k$-partition of~$G$.

  To prove the converse, let~$\mathcal{V} = \{V_i\}_{i \in [k]}$ be a
  connected $k$-partition of~$G$.  We assume without loss of
  generality that $w(V_i)\leq w(V_{i+1})$ for all~$i \in [k-1]$.
  Let~$(y,f)$ be a vector such that~$y=\rho(\mathcal{V})$
  and~$f=\tau(\mathcal{V})$.  For each~$i \in [k]$, every vertex
  in~$\vec T_i$ has in-degree at most one, and~$r_i$ is the smallest
  vertex in~$V(\vec T_i)$ with respect to the order~$\succ$.  Thus,
  inequalities~\eqref{flow:ineq:indeg} and~\eqref{flow:ineq:root} hold
  for~$(y,f)$.  From the definition of~$\rho(\mathcal{V})$, the entry
  of~$y$ indexed by~$(s,r_i)$ and~$i$ equals one, for all~$i \in [k]$.
  Consequently, $(y,f)$ also satisfies
  inequalities~\eqref{flow:ineq:source}.  Recall
  that~$g_i((s,r_i)) = |V_i|$ for every~$i \in [k]$.  This clearly
  implies that inequalities~\eqref{flow:ineq:capacity} are satisfied
  by~$(y,f)$.

  Note that, for every~$i \in [k]$, the function $g_i$ assigns to each
  arc~$(u,v) \in A(\vec T_i) \cup \{(s,r_i)\}$ the value one plus the
  sum of the sizes of the sub-arborescences of~$\vec T_i$ rooted at
  the out-neighbors of~$v$ in~$\vec T_i$.  Hence,
  inequalities~\eqref{flow:ineq:conservation2}
  and~\eqref{flow:ineq:conservation} hold for~$(y,f)$.  Finally,
  inequalities~\eqref{flow:ineq:asym} are satisfied, as we assumed
  that the elements of partition~$\mathcal{V}$ are in a non-decreasing
  order of weights.  Therefore, we conclude that~$(y,f)$ belongs
  to~$\mathcal{Q}_k(G,w)$.
\end{proof}


%% file: experiments_tmp.tex
\section{Computational experiments}
\label{section:experiments}

\subsection{Benchmark instances}
In order to compare the performance of  our algorithms with the exact algorithms that have been proposed in the literature~\cite{Mat14, ZhoWanDinHuSha19}, we ran our experiments on grid graphs and random connected graphs.
Our algorithms are based on the three formulations that we have described in the previous sections.

The grid instances have names in the format \textit{gg\_height\_width\_[a\textbar b]}, while random connected graphs instances have names in the format \textit{rnd\_n\_m\_[a\textbar b]}, where $n$ (resp. $m$) is the number of vertices
(resp. edges) of the graph. 
The characters (``a'' or ``b'') in the end of the name of an instance refer to the range of the weight distribution: character ``a'' (resp. ``b'') indicate that the weights are integers uniformly distributed in the closed interval~$[1, 100]$ (resp.  $[1, 500]$).

To create a random connected graph with $n$ vertices and $m$ edges, with $m > n - 1$, we first use Wilson's algorithm~\cite{Wilson96} to generate a uniformly random spanning tree~$T$ on $n$ vertices, and then add $m - n + 1$ new edges from $E(K_n) \setminus E(T)$ at random with uniform probability. 
Wilson's algorithm returns a spanning tree $T$ sampled from the set $\tau_n$ --- of all possible spanning
trees of $K_n$ ---  with probability ${1}/{|\tau_n|}$.

\subsection{Computational environment}

The computational experiments were carried out on a PC with Intel(R) Core(TM) i7-4720HQ CPU @ 2.60GHz, 4 threads, 8 GB RAM and Ubuntu 18.04.2 LTS. 
The code was written in C++ using Gurobi 8.1~\cite{Gurobi} and the graph library Lemon~\cite{Lemon}.
In order to evaluate strictly the performance of the described formulations, we deactivated all standard cuts used by Gurobi. 
Besides implementing the proposed formulations, we also implemented the Integer Linear Programming models introduced by Matic~\cite{Mat14} and Zhou~et~al~\cite{ZhoWanDinHuSha19}.

\subsection{Computational results}

The execution time limit was set to 1800 seconds. 
In the following tables, we show the number of explored nodes in the Branch-and-Bound tree (column ``Nodes'') and the time (in seconds) to solve the corresponding instance (column ``Time'').
If the time limit is reached, the table entry shows a dash (-).
Henceforth, when we refer to any of the formulations (or models) it
should be understood that we are referring to the corresponding exact algorithms that we
have implemented for them.
Thus, \CutForm\ refers to the Branch-and-Cut algorithm based on formulation~$\mathcal{C}_k$, and
\FlowForm\ and \FlowTwoForm\ refer to the Branch-and-Bound algorithms based on formulations~$\mathcal{F}_k$ and~$\mathcal{F}^\prime_k$. 
For the sake of simplicity, the names of these algorithms are shortened to Cut, Flow and Flow2.

In Table~\ref{table:facetable}, we show the impact of separating cross inequalities.
Columns ``CutF'' and ``Cut'' refer to \CutForm\ with and without the cross inequalities, respectively. 
Columns ``Conn. Cuts'' and ``Cross Cuts'' show the number of connectivity and cross inequalities
separated by the algorithm.  
We note that CutF was faster than Cut in most of the instances.
Furthermore, grids gg\_15\_15\_a and gg\_15\_15\_b
could only be solved by CutF.

Tables~\ref{table:gridtable} and~\ref{table:rndtable} indicate that the proposed algorithms substantially outperform the previous solution methods in the literature. 
On all grids instances, Flow had the best execution time. 
Furthermore, on grids with higher dimensions, the algorithms based on the formulations devised by Matic~\cite{Mat14} and Zhou et al.~\cite{ZhoWanDinHuSha19} could not find a solution within the time limit.
On the other hand, CutF and Flow were able to solve all the instances.

Considering the random graph instances, Table~\ref{table:rndtable} shows that, on some instances, \CutForm\ is better than the \FlowForm. 
Moreover, a reasonable amount of the instances were solved by Cut and Flow in the root node of the Branch-and-Bound tree.

Finally, Table~\ref{table:ktable} indicates that when the value of $k$ is greater than~$2$, the problem becomes much harder to solve. 
Only Flow  was able to solve gg\_07\_10\_a for $k = 3$ and $k = 4$. 
For $k = 5$ and $k = 6$,  none of the algorithms were able to solve the instance within the time limit.

\begin{table}[H]
\caption{Computational results for $\BCPtwo$ on grid graphs with the amount of cuts added by the algorithms.}
\small
\rowcolors{1}{lightgray}{}
\centering
\begin{tabular}{lrrrrrr}
\hiderowcolors
\toprule
& \multicolumn{2}{c}{Cut} & \multicolumn{3}{c}{CutF} \\ 
\cmidrule(r){2-3} \cmidrule(r){4-6}
Instance & Conn. Cuts & Time & Conn. Cuts & Cross Cuts & Time \\ 
\midrule
\noalign{\global\rownum=1}
\showrowcolors
gg\_05\_05\_a & 3242  & 0.73 & 577 & 796 & \textbf{0.17} \\ 
gg\_05\_05\_b & 3527  & 0.87 & 1043 & 1347 & \textbf{0.34} \\ 
gg\_05\_06\_a & 2696  & 0.47 & 58 & 576 & \textbf{0.11} \\ 
gg\_05\_06\_b & 5629  & 1.45 & 1255 & 1310 & \textbf{0.30} \\ 
gg\_05\_10\_a & 7544  & \textbf{1.80} & 10120 & 11436 & 3.70 \\ 
gg\_05\_10\_b & 13192 & 4.53 & 2556 & 3703 & \textbf{0.67} \\ 
gg\_05\_20\_a & 9802  & \textbf{3.86} & 46040 & 3976 & 47.57 \\ 
gg\_05\_20\_b & 36455 & \textbf{23.09} & 4185 & 46185 & 85.53 \\ 
gg\_07\_07\_a & 1074  & \textbf{2.9} & 5996 & 7538 & 3.18 \\ 
gg\_07\_07\_b & 10246 & 3.12 & 7328 & 8258 & \textbf{2.20} \\ 
gg\_07\_10\_a & 20397 & 8.42 & 14215 & 13767 & \textbf{5.82} \\ 
gg\_07\_10\_b & 58657 & 71.57 & 23294 & 2255 & \textbf{19.24} \\ 
gg\_10\_10\_a & 1178 & \textbf{4.62} & 23183 & 26620 & 17.08 \\ 
gg\_10\_10\_b & 316601 & 1313.07 & 26738 & 28558 & \textbf{22.45} \\ 
gg\_15\_15\_a & 0 & - & 83527 & 53287 & \textbf{121.56} \\ 
gg\_15\_15\_b & 0 & - & 116041 & 73920 & \textbf{250.01} \\ 
\bottomrule
\end{tabular}
\label{table:facetable}
\end{table}

\begin{table}
\caption{Computational results for $\BCPtwo$ on grid graphs.}
\small
\rowcolors{1}{lightgray}{}
\centering
\begin{tabular}{lrrrrrrrrrrrrrrr}
\hiderowcolors
\toprule
& \multicolumn{2}{c}{CutF} & \multicolumn{2}{c}{Flow} & \multicolumn{2}{c}{Flow2} & \multicolumn{2}{c}{Matic} & \multicolumn{2}{c}{Zhou} \\ 
\cmidrule(r){2-3} \cmidrule(r){4-5} \cmidrule(r){6-7} \cmidrule(r){8-9} \cmidrule(r){10-11} 
Instance & Nodes & Time & Nodes & Time & Nodes & Time & Nodes & Time & Nodes & Time \\ 
\midrule
\noalign{\global\rownum=1}
\showrowcolors
gg\_05\_05\_a & 252 & 0.17 & 299 & \textbf{0.06} & 1697 & 0.23 & 299 & 0.23 & 674 & 0.17 \\ 
gg\_05\_05\_b & 390 & 0.34 & 1815 & \textbf{0.12} & 1678 & 0.23 & 5957 & 0.74 & 1708 & 0.24 \\ 
gg\_05\_06\_a & 71 & 0.11 & 323 & \textbf{0.06} & 541 & 0.15 & 6229 & 1.24 & 2134 & 0.46 \\ 
gg\_05\_06\_b & 284 & 0.30 & 1843 & \textbf{0.13} & 202500 & 19.73 & 35627 & 7.49 & 1979 & 0.39 \\ 
gg\_05\_10\_a & 747 & 3.70 & 343 & \textbf{0.10} & 1931 & 1.30 & 74004 & 14.10 & 7240 & 2.84 \\ 
gg\_05\_10\_b & 300 & 0.67 & 511 & \textbf{0.13} & 1002 & 0.83 & 213650 & 43.46 & 19342 & 5.35 \\ 
gg\_05\_20\_a & 1403 & 47.57 & 328 & \textbf{0.16} & 2825 & 4.69 & - & - & 729662 & 218.24 \\ 
gg\_05\_20\_b & 2530 & 85.53 & 1700 & \textbf{0.71} & 396 & 1.07 & - & - & - & - \\ 
gg\_07\_07\_a & 959 & 3.18 & 311 & \textbf{0.11} & 411 & 0.58 & 103416 & 16.29 & 14807 & 4.73 \\ 
gg\_07\_07\_b & 615 & 2.20 & 1515 & \textbf{0.24} & 1033 & 1.18 & 483949 & 101.43 & 3255 & 1.43 \\ 
gg\_07\_10\_a & 779 & 5.82 & 449 & \textbf{0.19} & 391 & 0.74 & - & - & 11951 & 6.76 \\ 
gg\_07\_10\_b & 1479 & 19.24 & 606 & \textbf{0.16} & 872 & 1.48 & - & - & 18059 & 6.85 \\ 
gg\_10\_10\_a & 1111 & 17.08 & 313 & \textbf{0.18} & 262 & 1.11 & - & - & - & - \\ 
gg\_10\_10\_b & 1206 & 22.45 & 836 & \textbf{0.31} & 765832 & 595.67 & - & - & 3574970 & 957.98 \\ 
gg\_15\_15\_a & 1136 & 121.56 & 155 & \textbf{0.40} & 531 & 6.91 & - & - & - & - \\ 
gg\_15\_15\_b & 2562 & 250.01 & 1457 & \textbf{1.60} & - & - & - & - & - & - \\ 
\bottomrule
\end{tabular}
\label{table:gridtable}
\end{table}

\begin{table}
\caption{Computational results for $\BCPtwo$ on random graphs.}
\small
\rowcolors{1}{lightgray}{}
\centering
\begin{tabular}{lrrrrrrrrrr}
\hiderowcolors
\toprule
& \multicolumn{2}{c}{Cut} & \multicolumn{2}{c}{Flow} & \multicolumn{2}{c}{Flow2} & \multicolumn{2}{c}{Matic} & \multicolumn{2}{c}{Zhou} \\ 
\cmidrule(r){2-3} \cmidrule(r){4-5} \cmidrule(r){6-7} \cmidrule(r){8-9} \cmidrule(r){10-11}
Instance & Nodes & Time & Nodes & Time & Nodes & Time & Nodes & Time & Nodes & Time \\ 
\midrule
\noalign{\global\rownum=1}
\showrowcolors
rnd\_20\_30\_a & 7 & \textbf{0.02} & 463 & 0.05 & 1509 & 0.16 & 6410 & 0.52 & 163 & 0.08 \\ 
rnd\_20\_30\_b & 7 & \textbf{0.02} & 1729 & 0.10 & 7596 & 0.66 & 7249 & 0.91 & 939 & 0.11 \\ 
rnd\_20\_50\_a & 45 & \textbf{0.02} & 323 & 0.06 & 279 & 0.22 & 117 & 0.12 & 79 & 0.11 \\ 
rnd\_20\_50\_b & 639 & 0.13 & 311 & \textbf{0.07} & 299 & 0.11 & 1758 & 0.27 & 103 & 0.23 \\ 
rnd\_20\_100\_a & 1 & \textbf{0.01} & 1 & 0.02 & 169 & 0.19 & 79 & 0.08 & 77 & 0.20 \\ 
rnd\_20\_100\_b & 19 & \textbf{0.01} & 661 & 0.14 & 515 & 0.31 & 37 & 0.11 & 5 & 0.22 \\ 
rnd\_30\_50\_a & 182 & 0.12 & 23 & \textbf{0.04} & 574 & 0.15 & 2793 & 0.61 & 594 & 0.19 \\ 
rnd\_30\_50\_b & 55 & \textbf{0.08} & 5499 & 0.44 & 997 & 0.28 & 3169 & 0.62 & 1051 & 0.22 \\ 
rnd\_30\_70\_a & 1 & \textbf{0.02} & 1 & 0.03 & 295 & 0.19 & 1608 & 0.38 & 203 & 0.37 \\ 
rnd\_30\_70\_b & 15 & \textbf{0.03} & 1326 & 0.15 & 355 & 0.21 & 2721 & 0.84 & 528 & 0.34 \\ 
rnd\_30\_200\_a & 1 & \textbf{0.01} & 1 & 0.06 & 49 & 0.34 & 1 & 0.11 & 19 & 0.46 \\ 
rnd\_30\_200\_b & 1 & \textbf{0.01} & 57 & 0.17 & 327 & 0.99 & 2133 & 1.88 & 141 & 0.48 \\ 
rnd\_50\_70\_a & 75 & 0.16 & 23 & \textbf{0.06} & 1026 & 0.41 & 3812 & 1.42 & 1270 & 0.40 \\ 
rnd\_50\_70\_b & 474 & 0.67 & 787 & \textbf{0.10} & 2606 & 0.59 & 2530 & 1.05 & 1605 & 0.61 \\ 
rnd\_50\_100\_a & 83 & 0.17 & 327 & \textbf{0.09} & 462 & 0.76 & 7548 & 1.78 & 180 & 0.37 \\ 
rnd\_50\_100\_b & 147 & 0.21 & 15 & \textbf{0.05} & 468 & 0.64 & 5043 & 1.62 & 708 & 0.51 \\ 
rnd\_50\_400\_a & 1 & \textbf{0.03} & 1 & 0.13 & 1 & 0.37 & 2525 & 2.91 & 99 & 1.47 \\ 
rnd\_50\_400\_b & 1 & \textbf{0.03} & 1 & 0.09 & 478 & 3.22 & 2795 & 2.73 & 41 & 2.21 \\ 
rnd\_70\_100\_a & 55 & 0.22 & 583 & \textbf{0.11} & 865 & 0.60 & 401 & 0.12 & 915 & 0.68 \\ 
rnd\_70\_100\_b & 755 & 1.28 & 1112 & \textbf{0.19} & 1084 & 0.60 & 5173 & 2.12 & 1696 & 1.03 \\ 
rnd\_70\_200\_a & 1 & \textbf{0.05} & 71 & 0.18 & 1 & 0.43 & 2235 & 1.15 & 385 & 1.53 \\ 
rnd\_70\_200\_b & 21 & \textbf{0.09} & 1034 & 0.26 & 42 & 0.62 & 9756 & 2.26 & 267 & 0.62 \\ 
rnd\_70\_600\_a & 1 & \textbf{0.01} & 1 & 0.15 & 35 & 2.05 & 57 & 0.37 & 1 & 0.70 \\ 
rnd\_70\_600\_b & 19 & \textbf{0.07} & 685 & 0.54 & 1 & 0.80 & - & - & 146 & 3.01 \\
rnd\_100\_150\_a & 235 & 2.27 & 71 & \textbf{0.15} & 1718 & 2.89 & - & - & 1149 & 1.24 \\ 
rnd\_100\_150\_b & 63 & 0.41 & 539 & \textbf{0.19} & 527 & 1.16 & 1534 & 1.40 & 842 & 1.20 \\ 
rnd\_100\_300\_a & 1 & \textbf{0.06} & 1 & 0.13 & 1381 & 4.89 & 28994 & 6.17 & 490 & 1.56 \\ 
rnd\_100\_300\_b & 19 & \textbf{0.10} & 252 & 0.20 & 1 & 1.21 & 28837 & 10.39 & 343 & 2.62 \\ 
rnd\_100\_800\_a & 1 & \textbf{0.04} & 1 & 0.24 & 1 & 1.51 & 105575 & 112.73 & 1 & 1.52 \\ 
rnd\_100\_800\_b & 1 & \textbf{0.05} & 35 & 0.49 & 1 & 1.80 & - & - & 420 & 5.28 \\ 
rnd\_200\_300\_a & 353 & 12.05 & 1 & \textbf{0.20} & 245114 & 147.72 & 358170 & 445.98 & 1917 & 18.27 \\ 
rnd\_200\_300\_b & 1618 & 29.48 & 879 & \textbf{0.46} & 2765 & 12.60 & 49330 & 67.92 & 4606 & 9.69 \\ 
rnd\_200\_600\_a & 1 & \textbf{0.12} & 1 & 0.27 & 2965 & 20.61 & - & - & 735 & 15.08 \\ 
rnd\_200\_600\_b & 39 & 1.14 & 1295 & \textbf{0.93} & 5169 & 31.24 & - & - & 939 & 10.25 \\ 
rnd\_200\_1500\_a & 1 & \textbf{0.08} & 1 & 0.60 & 1 & 4.67 & 24900 & 83.52 & 1 & 3.25 \\ 
rnd\_200\_1500\_b & 1 & \textbf{0.08} & 1 & 0.51 & 7956 & 160.97 & 32658 & 103.12 & 489 & 17.31 \\ 
rnd\_300\_500\_a & 227 & 4.88 & 675 & \textbf{0.79} & 4955 & 30.65 & 5848 & 12.15 & - & - \\ 
rnd\_300\_500\_b & 827 & 19.65 & 201 & \textbf{0.48} & 3165 & 29.86 & 6917 & 14.28 & 1635 & 33.06 \\ 
rnd\_300\_1000\_a & 1 & \textbf{0.23} & 1726 & 1.69 & 3580 & 48.44 & 7936 & 26.08 & 316 & 16.01 \\ 
rnd\_300\_1000\_b & 1 & \textbf{0.55} & 1353 & 2.03 & 4633 & 61.26 & - & - & 642 & 16.74 \\ 
rnd\_300\_2000\_a & 1 & \textbf{0.10} & 1 & 0.85 & 4404 & 132.81 & 12852 & 87.30 & 34 & 58.24 \\ 
rnd\_300\_2000\_b & 1 & \textbf{0.16} & 86 & 1.51 & 2541 & 69.61 & - & - & 41 & 31.19 \\ 
\bottomrule
\end{tabular}
\label{table:rndtable}
\end{table}

\begin{table}[H]
\caption{Computational results for $\BCPk$ on a grid with height 7 and width 10.}
\rowcolors{1}{lightgray}{}
\centering
\begin{tabular}{lcrrrrrrrr}
\hiderowcolors
\toprule
& & \multicolumn{2}{c}{CutF} & \multicolumn{2}{c}{Flow} & \multicolumn{2}{c}{Flow2} & \multicolumn{2}{c}{Zhou} \\ 
\cmidrule(r){3-4} \cmidrule(r){5-6} \cmidrule(r){7-8} \cmidrule(r){9-10} 
Instance & $k$ & Nodes & Time & Nodes & Time & Nodes & Time & Nodes & Time \\ 
\midrule
\noalign{\global\rownum=1}
\showrowcolors
gg\_07\_10\_a & 3 & - & - & 1010 & \textbf{0.29} & - & - & - & - \\ 
gg\_07\_10\_a & 4 & - & - & 650619 & \textbf{82.32} & - & - & - & - \\ 
gg\_07\_10\_a & 5 & - & - & - & - & - & - & - & - \\ 
gg\_07\_10\_a & 6 & - & - & - & - & - & - & - & - \\ 
\bottomrule
\end{tabular}
\label{table:ktable}
\end{table}

%% file: conclusion.tex

\section{Conclusion}
\label{section:conclusion}

We proposed three mixed integer linear programming formulations for the Balanced Connected $k$-Partition Problem. 
To avoid some symmetries, our formulations impose an ordering of the classes $\{V_i\}_{i \in [k]}$, such that $w(V_i) \leq w(V_{i + 1})$, for all $i \in [k - 1]$. 
The first one, $\mathcal{C}_k$, is defined on the input graph (differently from the other two formulations) and has a potentially large (i.e. exponential) number of connectivity inequalities. 
Moreover, we also presented a new class of valid inequalities for this formulation, and separate them on planar graphs. 
Computational experiments indicated that the addition of these inequalities improves greatly the performance of the algorithm.
In the case the vertices have the same weight, we proved that the associated polytope is full-dimensional and characterized several inequalities that define facets of this polytope.

We also proposed two formulations based on flows in a digraph.
Formulation~$\mathcal{F}_k$ has a polynomial number of variables and constraints. 
To avoid symmetrical solutions and dependency on the weights of the vertices, we designed formulation~$\mathcal{F}^\prime_k$.
However, in our computational experiments, the performance of~\FlowForm\ was always superior to \FlowTwoForm.

Preliminary experiments  showed that both \CutForm\ and \FlowForm\ have better performance when compared with the exact algorithms in the literature.
We plan to carry out further experiments on more instances, specially on some
real data (e.g. problems on police patrolling), to evaluate the performance of all algorithms mentioned in this work.